\documentclass[10pt,a4paper]{article}
\usepackage[a4paper]{geometry}

\usepackage{changepage}

\usepackage[utf8]{inputenc}

\usepackage{textcomp,marvosym}

\usepackage{fixltx2e}

\usepackage{amsmath,amssymb}
\usepackage{amsthm}
\usepackage{mathtools}
\usepackage{enumerate}
\usepackage[mathscr]{eucal}
\usepackage{siunitx}
\usepackage{hyperref}

\usepackage{cite}

\usepackage{graphicx}
\usepackage{subfigure}
\usepackage{epstopdf}
\usepackage{tikz}
\usepackage{pgfplots}

\newcommand*\samethanks[1][\value{footnote}]{\footnotemark[#1]}
 \newtheorem{thm}{Theorem}

 \newtheorem{alg}[thm]{Algorithm} 

\newtheorem{defn}[thm]{Definition}
\newtheorem{rem}{Remark}

\renewcommand{\d}{\mathrm{d}}

\newcommand {\E} { {\mathbb E} }

\renewcommand{\H}{\mathrm{H}}
\newcommand{\Var}{\mathrm{Var}}
\newcommand{\dps}{\displaystyle}

\newcommand{\cmsq}{cm\textsuperscript{2}}
\newcommand{\sinv}{min\textsuperscript{-1}}
\newcommand{\VEGF}{\mathrm{VEGF}}
\newcommand{\q}{\mathrm{p53}}

\title{Variance-reduced simulation of stochastic
agent-based models for tumor growth}
\author{Annelies Lejon\thanks{Department of Computer Science, KU Leuven, Celestijnenlaan
200A, 3001 Leuven, Belgium (firstname.lastname@cs.kuleuven.be). The first
author's work was supported by the Agency for Innovation by Science and
Technology in Flanders (IWT)}\and Bert Mortier\samethanks{} \and Giovanni Samaey\samethanks{}}
\begin{document}
\maketitle

\begin{abstract}
We investigate a hybrid PDE/Monte Carlo technique for the variance reduced simulation of an agent-based multiscale model for tumor growth. The variance reduction is achieved by combining a  simulation of the stochastic agent-based model on the microscopic scale with a deterministic solution of a simplified (coarse) partial differential equation (PDE) on the macroscopic scale as a control variable. We show that this technique is able to significantly reduce the variance with only the (limited) additional computational cost associated with the deterministic solution of the coarse PDE. We illustrate the performance with numerical experiments in different regimes, both in the avascular and vascular stage of tumor growth.
\end{abstract}

\section{Introduction}
Tumor growth is a complex biological
phenomenon consisting of processes on different scales. On the cellular level -- which will be referred to as the microscopic scale in this paper -- one has to track the random motion of cells, as well as the cell division and cell death. The latter are governed by numerous intracellular processes. Furthermore, the
cellular behavior is strongly coupled to the environment and vice versa.
For example, cell proliferation is determined by the local oxygen concentration and the local cell density while hypoxia on the other hand can
trigger apoptosis, but cells also consume oxygen. This two-way feedback creates
a very specific dynamics characterizing the development of the tumor. A hypoxic
zone develops in the middle of the tumor, which in turn triggers
endothelial cells to vascularize the tumor. This process, also known as
angiogenesis~\cite{carmeliet2005}, ensures that the tumor's need for oxygen and other nutrients is satisfied, which implies that
the tumor can grow further.\\ Smaller avascular tumors can be easily simulated on the microscopic scale using agent-based models. We can distinguish two classes of models. On one hand, cellular automata update grid cells based on a number of
phenomenological rules~\cite{Kavousanakis2012,Olsen2013}, while on the
other hand lattice-free models typically consist of a set of ordinary differential equations (ODEs) attached to each cell.\\
On long time scales, we are typically interested in the tumor as a whole, which we call the macroscopic scale. Agent-based models are typically not well suited to use on this larger scales, since the individual based character implies a large computational cost for a large number of particles, corresponding to larger tumors. One may choose to model the system directly on this
scale using continuum models, based on mass balance
equations~\cite{Berger2011,roose2007, Wise2008,stein2007,Spill2014,Hubbard2013}.
While this approach is significantly cheaper and easier to analyze than agent-based models, it cannot capture discrete features as branching of a vascular network or events regulated by intracellular concentrations. This insight gave rise to multiscale models where agent-based models are typically used to model the cellular component, while the environment is mostly described by a set of reaction-diffusion partial differential equations (PDEs),
corresponding to the macroscopic scale. Examples can be found in
~\cite{Owen2011,D'Antonio2012,Hanahan2011,Anderson2006,Owen2009}.
For a review about the current state of the art in multiscale-modeling of
tumor growth, we refer to~\cite{Deisboeck2011}.\\
Due to the random motion and the influence on the environment, the simulations are subject to noise. When simulating with a standard Monte Carlo algorithm, the variance can only be reduced by increasing the number of particles at the cost of computational efficiently. Various techniques for variance reduction such as antithetic variables, control variates and importance sampling are described in literature, see e.g.~\cite{caflisch1998,Lemieux2009} for an overview. Recently, several hybrid PDE/Monte Carlo algorithms have been proposed in the literature to achieve variance reduction by coupling a PDE-based discretization to a Monte Carlo simulation~\cite{dimarco2007, Rousset2013, radtke2009}.\\
The contributions of this paper are two-fold:
\begin{itemize}
\item We develop a multiscale model where the random motion is modeled
using stochastic differential equations (SDEs), the intracellular variables for
the cell cycle and apoptosis are described by ODEs and the environment,
consisting of diffusible components, is modeled by PDEs. The model is a modified
version of the cellular automaton model of Owen et al.~\cite{Owen2011}. The main
differences are that the new model is lattice-free and the fact that our model does not contain any explicit delay terms.
\item We propose a novel technique to reduce the variance on the results of the
adapted multiscale agent-based model. Based on the ideas in~\cite{Rousset2013, dimarco2007}, we develop a hybrid PDE/Monte Carlo method using a coarse stochastic process (called the control process)
and a corresponding PDE. In this specific case, the control process modeling the spatial behavior of the individual cells contains all details of the microscopic model except for cell births, cell
deaths and VEGF secretion. The keypoint is to obtain this missing information with reduced variance by an appropriate coupling between the full microscopic agent-based model and the control process.
\end{itemize}
We first give a detailed overview of the different layers of the model. Next, we describe the variance reduction
algorithm in the section~\ref{sec:varred}. We illustrate the technique
numerically in section~\ref{sec:results}. Finally, in section~\ref{sec:disc} we elaborate on a few possibilities for future research.

\section{Models\label{sec:models}}
In this section, we describe a multiscale model for tumor growth. 
The microscopic model for tumor growth is based on the ideas used to describe
bacterial chemotaxis~\cite{Erban2004,Rousset2013}, the multiscale cellular
automaton model was developed by Owen and coworkers~\cite{Owen2011} and our goal is
to perform variance reduction in order to estimate the resulting population
densities in a more accurate way. As in~\cite{Rousset2013}, the
model proposed in this paper is time and space continuous. Apart from the fact that the model is lattice-free, making the computational cost quasi independent of the size of the domain and hence, we can easily
rescale the system to simulate larger tumors (compared to the examples given in~\cite{Owen2011}).\\
 We distinguish two main components:
the environment, modeled by a couple of reaction diffusion equations and the
agent-based model describing the individual cellular motion and internal variables (e.g. cell cycle,
apoptosis state and internal concentrations such as VEGF and p53) attached to each cell.\\
We consider three types of cells, indexed by $1\le p \le P=3$: normal cells
($p=1$), cancer cells ($p=2$), and endothelial cells (that build up blood
vessels, $p=3$).
For each of these cell types, we consider an ensemble of $I_p(t)$ cells,
and consider three state variables: position $x\in \mathbb{R}^2$, cell cycle
phase $\phi\in[0,1]$. The intracellular concentrations $[\VEGF]_\mathrm{int}$, $[p53]$ and apoptosis variable $z$ are scalars $\in \mathbb{R}$. 
These cells evolve according to evolution laws that depend on the concentration
$[O_2](x,t)$ of oxygen and $[\VEGF](x,t)$ of the Vascular Endothelial Growth
Factor (which we call the environment).\\
Remark that the reaction-diffusion partial differential equations (PDEs)
describing the diffusible components of the environment still need to be solved
on a grid, but this cost is marginal due to the sparsity of the involved linear
systems, which ensures that the cost dependent on the domain size is limited.\\
We now give an overview of the notations that will be used throughout the
paper, after which  we describe the evolution laws for the environment, and
detail the evolution laws for each of the cell types. The cell type dependency is
mainly caused by cell type dependent coefficients, which will be discussed
later on in the section describing the agent-based model in more detail.
\paragraph{Notation}
\begin{itemize}
  \item The state variables attached to a single cell of type $p$ at time $t$ 
  are position $X_{p}(t)$, cell cycle phase $\Phi_{p}(t)$, generation
  $\zeta_{p}(t)$, internal concentrations $[\q]_p(t)$, $[\VEGF_\mathrm{int}]_p(t)$ and apoptosis variable $Z_{p}(t)$.
  \item Particle number densities are denoted by $n_p(x,t)$ indexed by a
  suitable subscript to indicate the nature of the density. Further,
  $n_v(x,t)$ is used to describe the vascular density.
  \item To keep a consistent notation throughout the paper, we introduce the following convention. If, at a moment $t=t^\star$,
  the cell with index $i^\star$ in population $p$ divides, we set 
  \begin{equation}
  	I_p(t^\star) = I_p(t^\star_{-})+1\label{eq:Ipbirth}
  \end{equation}
  in which the symbol $t^\star_{-}$ is used to emphasize that the involved number of cells is meant to be taken just before the division.
  Simultaneously, we introduce a new cell as specified in the paragraph concerning cell division (see page \pageref{par:celldiv}). When a cell undergoes apoptosis, it is removed from the simulation.
  To avoid cumbersome renumbering of the cells in the text, we associate a weight $w_{i,p}(t)$ to each of the cells. If the cell is alive, the corresponding weight is one;
  upon apoptosis, it becomes zero. The active number of cells is therefore:
  \begin{equation}
  	\bar{I}_p(t) = \sum_{i=1}^{I_p(t)}w_{i,p}(t)
  \end{equation}
  \item The evolution of the state variables is influenced in various ways by
  the (local) environment. The latter will be modelled by means of diffusible
  components $[\VEGF](x,t)$ describing the VEGF concentration, while
  $[O_2](x,t)$ denotes the oxygen concentration.
\end{itemize}

\subsection {Agent-based model\label{subsec:abm}}
In this section we give a detailed overview of the evolution of the different
state variables attached to each cell of the different cellular populations.
 \paragraph{Position.}The random motion of the
position of the cells is described as a biased Brownian motion. Cells of type
$p$ move randomly with diffusion coefficient $D_p$, and the cells are
chemotactically attracted towards high concentrations with sensitivity $\chi_p$.
This sensitivity is especially important for the endothelial cells, responsible for blood vessel growth,
\begin{equation}
  dX_{p}(t) = \chi_p\nabla
  [VEGF](X_p(t),t))\left(1-\dfrac{n_p(X_p(t),t)}{n_{\mathrm{max},p}}\right)\d
  t+\sqrt{2D_p}\d W_t 
\label{eq:pos_disc}
\end{equation}
 The cell number density $n_p$, can be computed as:
\begin{equation}
    n_p(x,t) = \sum_{i=1}^{I_{p}(t)}w_{i,p}(t)\delta_{X_{i,p}(t)}
    \label{eq:number_density}
\end{equation}
where $I_{p}(t)$ denotes the total number of cells of type $p$ at time $t$ and $\delta$ denotes the classical Dirac kernel, resulting in a standard histogram.
Remark, in contrast to the cellular automaton model described in~\cite{Owen2011},
the resulting equation for the position is a stochastic differential equation
(SDE) instead of the discrete space jumps used in~\cite{Owen2011}. Finally, we
have to stress the fact that the above equation is general for all the cell
types. To be more concrete, normal cells don't move at all, while cancer cells
are characterized by pure diffusive motion and endothelial cells demonstrate
diffusive behavior but they also respond to chemotactic cues.
 \paragraph{Cell
division\label{par:celldiv}} is modeled by means of the following ODE:
\begin{equation}
  \dfrac{\d \Phi_{p}(t)}{\d t} =
  \dfrac{[O_2](X_p(t),t)}{\tau_{\mathrm{min},p}(C_{\phi,p}+[O_2](X_p(t),t))}\H
   \left(\zeta_{p}(t)-\zeta_{p,\max}\right)
\end{equation}
  
where $\tau_{\min,p}$ denotes the minimal time needed for a cell to
complete one cell cycle and $\zeta$ indicates the generation of a cell. Remark
that $\tau_{\mathrm{min},p}$ depends on the cell type. To be more specific, cancer cells are able to proceed twice as
fast as normal cells during the cell cycle in a given environment (see
table~\ref{tab:populations}). Naturally the cell cycle speed depends on the
local oxygen concentration $[O_2](X_p(t),t)$ as observed by the cell while
evolving through the cycle.
The higher the oxygen concentration, the faster the cycle is completed, while
the cell cycle is put on hold when the cell suffers from hypoxia. A more
detailed biological motivation for this model can be found in
\cite{Owen2011,Tyson2001} and its supplementary material. Remark that in the cellular automaton model by Owen et
al (see~\cite{Owen2011}) all cells can divide an unlimited number of
times, which corresponds to the hypothesis that all cells are stem cells, which
is obviously not a realistic assumption. 
Thus, we have extended the model to account for the fact that cells are only
able to divide a finite number of times (i.e. $\zeta_{p,\max}$). To be more specific we added a factor $\H(\zeta_p(t)-\zeta_{p,\max})$ to check for the generation of the
corresponding cells. Here, we assume that normal cells can divide only $4$
times, which is consistent with~\cite{Fletcher2012}.
On the other hand we assume that all the cancer cells are cancer stem cells,
which is still a simplification.\\
If, for the cell with index $i^\star$ in population $p$ at time $t=t^\star$, we obtain $\Phi(t^\star)\geq 1$, we introduce a new cell in the simulation.
We adjust $I_p(t)$ according to equation~\eqref{eq:Ipbirth} and set $\Phi_{i^\star,p}(t)=0$. and the generation of the parent cell increases by one. The new cell inherits the state from the cell that divides except for the generation $\zeta$:
\begin{equation}
	\begin{aligned}
	X_{I_p(t),p}(t^\star) &= X_{i^\star,p}(t^\star) \\
	[\q]_{I_p(t),p}(t^\star)  & = [\q]_{i^\star,p}(t^\star) \\
	\Phi_{I_p(t),p}(t^\star) &= \Phi_{i^\star,p}(t^\star)\\
\end{aligned} 
\qquad
\begin{aligned}
 Z_{I_p(t),p}(t^\star)  &= Z_{i^\star,p}(t^\star)\\
[\VEGF_\mathrm{int}]_{I_p(t),p}(t^\star)&= [\VEGF_\mathrm{int}]_{i^\star,p}(t^\star) \\
 \zeta_{I_p(t),p}(t^\star) &= 0\\
\end{aligned}
\end{equation}

\paragraph{Intracellular model.} We introduce a intracellular module consistent
with \cite{Owen2011} in order to describe some important intracellular
concentrations, namely the p53 concentration $[p53]$  and the intracellular VEGF concentration $[\VEGF_\mathrm{int}]$. The former
can be seen as an estimator for the number of mutations that a cell has
undergone during its lifetime. We have:
\begin{equation}
\begin{aligned}
    \dps\dfrac{\d [\q]_p(t)}{\d t} & = c_1 -c_2
    \dfrac{[O_2](X_p(t),t)}{C_{p53}+[O_2](X_p(t),t)}[\q]_p(t)\\[8pt] 
    \dps \dfrac{\d [\mathrm{VEGF}_\mathrm{int}]_p(t)}{\d t} &= c_3
    -c_4\dfrac{[\q]_p(t)[\VEGF_\mathrm{int}]_p(t)}{J_5+[\VEGF_\mathrm{int}]_p(t)}\\
    \quad&+c_5\dfrac{[O_2](X_p(t),t)}{C_\mathrm{VEGF}+[O_2](X_p(t),t)}[\VEGF_\mathrm{int}]_p(t)
\end{aligned}
\end{equation}
Cells are storing VEGF intracellular (i.e. [VEGF]\textsubscript{int}) during hypoxic
conditions and release it once this intracellular concentration has reached a certain
threshold level $[\VEGF_\mathrm{int}]_\mathrm{thr}$.
Further, $c_1,\ldots c_5$ and $C_{\mathrm{p53}},C_\mathrm{VEGF}$ are constants that can be found
in table~\ref{tab:populations}. 
Next we describe the model for apoptosis, depending on the cell type. Therefore we formally define $\gamma_{\mathrm{apt},p}(z,n_p)=F_p(z,n_p)$ as the apoptosis rate, which is further specified in the following paragraphs.
\paragraph{Apoptosis for normal cells.} For normal cells, cell death is
completely determined by the subcellular p53-concentration. So, we set the
apoptosis variable $z:=[p53]$. The apoptosis threshold $\gamma_\mathrm{apt}$ can
then then be written as:
\begin{equation}
    \gamma_{\mathrm{apt},1}(z,n_1)= \H\left(z-
    z_\mathrm{high}\H(n_\mathrm{thr}-n_\mathrm{1}) -
    z_\mathrm{low}\H(n_\mathrm{1}-n_\mathrm{thr})\right)\label{eq:death_normal}
\end{equation}
where $\mathrm{H}$ indicates the Heaviside function.
 This definition of $\gamma_\mathrm{apt}$
implies that normal cells undergo apoptosis if $\gamma_\mathrm{apt}(z,n_1)=1$
corresponding to the situation that $z$ has reached a certain threshold value
depending on the harshness of the environment. The threshold value is lower in case of a harsh environment, defined as $n_1<n_\mathrm{thr}$, where $n_\mathrm{thr}$ denotes a threshold
value for the normal cells.
\paragraph{Apoptosis for cancer cells.} In contrast to normal cells, the
apoptosis mechanism for tumor cells is independent of the p53-concentration since this
mechanism to regulate the normal cell cycle does not function properly anymore in a tumor. Cancer cells are able to go into a quiescent state when expressed to hypoxic circumstances, meaning that they don't consume any nutrients anymore for a 
while. However the duration of this quiescent state is limited, which implies
that cancer cells will also undergo apoptosis when the hypoxia holds too long.
On the other hand, cancer cells have the ability to recover quickly once there is
again more oxygen available. This mechanism can be modeled by the following
equation:

\[
    \dfrac{\d Z(t)}{\d t} =
    \underbrace{A\H([O_2]_\mathrm{thr}-[O_2](X_{p}(t),t))}_{\text{Linear
    increase during hypoxia}}
    -\underbrace{BZ(t)\H([O_2](X_{p}(t),t)-[O_2]_\mathrm{thr})}_{\text{Exponential
    decay if }[O_2](X_p(t),t)>[O_2]_\mathrm{thr}}
\]
where $A,B$ are constants. Further, the first term models the hypoxic state,
i.e. the local oxygen concentration $[O_2](X_p(t),t)$ drops below the
threshold level $[O_2]_\mathrm{thr}$. During this hypoxic period, the internal
variable $z$ increases steadily.
On the other hand, the second term describes the recovery of the cancer cells if
the environment is not hypoxic anymore, which is captured by the exponential
decay term of $Z(t)$. Cancer cells die if $Z(t)\geq 1$, corresponding to
$\gamma_{\mathrm{apt},2}(z)=\H(z-1)$.
\paragraph{Endothelial Cells.} Remark that the model equations concerning cell
division and cell death will not be used for endothelial cells. Consistent with
existing literature, the so-called \textit{snail-trail} approach is used to
model sprouting angiogenesis~\cite{Bentley2008,Hellstrom2007}.\\
\paragraph{Full agent-based model.} This results in the following set of equations for the full agent-based model:
\begin{equation}
	\begin{cases}
		  dX_{p}(t) = \chi_p\nabla
  [VEGF](X_p(t),t))\left(1-\dfrac{n_p(X_p(t),t)}{n_{\mathrm{max},p}}\right)\d
  t+\sqrt{2D_p}\d W_t \\[12pt]
    \dfrac{\d \Phi_{p}(t)}{\d t} =
  \dfrac{[O_2](X_p(t),t)}{\tau_{\mathrm{min},p}(C_{\phi,p}+[O_2](X_p(t),t))}\H
   \left(\zeta_{p}(t)-\zeta_{p,\max}\right)\\[14pt]
    \dps\dfrac{\d [\q]_p(t)}{\d t} = c_1 -c_2 \dfrac{[O_2](X_p(t),t)}{C_{p53}+[O_2](X_p(t),t)}[\q]_p(t)\\[14pt]
    \dps \dfrac{\d [\mathrm{VEGF}_\mathrm{int}]_p(t)}{\d t} = c_3
    -c_4\dfrac{[\q]_p(t)[\VEGF_\mathrm{int}]_p(t)}{J_5+[\VEGF_\mathrm{int}]_p(t)}+c_5
    \dfrac{[O_2](X_p(t),t)}{C_\mathrm{VEGF}+[O_2](X_p(t),t)}[\VEGF_\mathrm{int}]_p(t)\\[14pt]
\gamma_{\mathrm{apt},p}(z,n_p,t) = F_p(z,n_p,t)
	\end{cases}\label{eq:abm_model}
\end{equation}
The corresponding parameter values can be found in table~\ref{tab:populations}.
\begin{table}[!h]
\centering
\begin{tabular}{|c|c|c|c|c|}
\hline
Parameter & $n_1$ & $n_2$ & $n_3$ & units
\\ \hline
$\chi_p$ & $0.0$ & $0.0$ & $\num{2e-4}$ &
\cmsq/min/nM\\
$\zeta_{p,\max}$ & $4$ & $\infty$ & $4$ & times \\
 $C_{\phi,p}$ &  $3$& $1.4$ & & $\mathrm{mmHg}$\\
 $C_{\VEGF}$ & 	$0.01$& $0.01$ & $0.01$& mmHg\\
 $C_{\q}$ & $0.01$& $0.01$ & $0.01$& mmHg\\
 $\tau_{p,\min}$ & $\num{3000}$ & $\num{1600}$ & & min\\
 $z_\mathrm{high}$ & $0.8$ & & & dimensionless\\
 $z_\mathrm{low}$ & $0.08$ & & & dimensionless\\
 $n_\mathrm{thr}$ & $0.75$ & &  & dimensionless\\
$[O_2]_\mathrm{thr}$ & & $8.9$ & &  mmHg\\ 
$[\VEGF_\mathrm{int}]_\mathrm{thr}$ & $0.27$ & $0.27$& & nM\\
 $c_1$ & $\num{ 2e-3}$ & $\num{ 2e-3}$& & \sinv\\
 $c_2$ & $\num{1e-2}$ &  $\num{1e-2}$& & \sinv\\
$c_3$ & $\num{2e-3}$ & $\num{ 2e-3}$& & \sinv\\
$c_4$ & $\num{2e-3}$ & $\num{ 2e-3}$& & \sinv\\
$c_5$ & $\num{ 1e-2}$ & $\num{ 1e-2}$& & \sinv\\
$J_5$ & $0.04$ & $0.04$ & & nM\\
$A$& & $1$ & & \sinv\\
$B$& & $\num{2.5e-3}$ & &\sinv\\\hline
\end{tabular}
\caption{Parameter values related to the populations.\label{tab:populations}}
\end{table}

\subsection{Coarse Description\label{subsec:macro}}
An alternative approach to model tumor growth is to describe the evolution of the
populations as a whole in a probabilistic way using partial differential
equations (PDEs). In general, this approach yields a reaction-diffusion PDE.
However, in this case it is not possible to derive a closed formulation for
the reaction terms since those all depend on intracellular variables. A
continuum description for the model outlined above without births and deaths can
be found in~\cite{Hillen2009}. The resulting macroscopic equation -- achieved by taking the limit for a high number of particles -- for the evolution of
the populations reads:
\begin{equation}
    \partial_t n_p(x,t) =D_p\nabla^2
    n_p(x,t)-\chi_p\nabla.\left[n_p(x,t)\left(1-\dfrac{n_p(x,t))}{n_{p,\max}}\right)\nabla
    [VEGF](x,t)\right]\label{eq:macro}
\end{equation}
where no reactions (cell divisions, cell deaths) are taken into account.
Next, we introduce the following macroscopic time-stepper:
\begin{equation}
\begin{aligned}
n_p(x,t^{k+1}) &=  n_p(x,t^k)+\delta t D_p\nabla^2
n_p(x,t^k)\\
&\quad-\delta t\left(\chi_p\nabla.\left[n_p(x,t^k)\left(1-\dfrac{n_p(x,t^k))}{n_{p,\max}}\right)\nabla
[VEGF](x,t^k)\right]\right)
\end{aligned}
\end{equation}
 which uses a first order
Euler discretization to discretize the time derivative and a second order
central finite volume scheme to discretize the spatial derivatives. Further
details can be found in section~\ref{sec:results}.
\subsection{Angiogenesis\label{subsec:Angio}}
The growth of new blood vessels, also known as \textit{angiogenesis} is
essential for the development of a tumor. Hanahan and co-authors identified it as one of the hallmarks of cancer (see \cite{Hanahan2011, Hanahan2000}).
In the early stages of cancer, the existing vasculature is able to provide enough oxygen and
other nutrients. But as soon as the size of the tumor has reached a certain threshold, a hypoxic zone develops in the middle of the tumor.
To cope with this phenomenon, the tumor secretes VEGF, a growth factor, which triggers endothelial cells to move
chemotactically towards this hypoxic zone and grow new blood vessels.  In this paper we will use an existing model for
angiogenesis, described in \cite{Owen2011}. We distinguish two phenotypes:
endothelial cells can either be motile leader cells (also called \textit{tip}
cells) or static \textit{stalk} cells. We model proliferation of endothelial
cells by means of the so-called \textit{snail-trail} approach, where each tip
cell produces a new (static) endothelial cell at its previous position, creating a trail of static \textit{stalk} cells behind him. Apart from
this feature, new tip cells --known as \textit{sprouts}-- can emerge from active
vessels with sprouting probability $P$\textsubscript{sprout} along the active
vessels. (see~\cite{Owen2011}):
\begin{equation}
	P_\mathrm{sprout} = \delta t \dfrac{P_{\max}[VEGF](x,t)}{V_\mathrm{sprout}+[VEGF](x,t)}
\end{equation}
where $P_{\max}=\num{3e-4}$~\sinv~indicates the maximal endothelial sprouting
rate (~see~\cite{Owen2011}) and $V$\textsubscript{sprout}$=0.5$~nM denotes the VEGF concentration at which the sprouting probability is half maximal. Remark that the probability of the emergence of two
sprouts close to each other within the same timestep is zero.  A biological explanation  for this fact can be found in the supplementary material provided with \cite{Owen2011} and~\cite{Gerhardt2008,Bentley2008,Jakobsson2010}, where
the authors pointed out that delta-notch signaling inhibits the formation of new sprouts in neighbouring
endothelial cells. Additionally, the vessel radii are also adapted dynamically based on the work of~\cite{Owen2011} where pruning
of the vessels was also incorporated in the model: if the pressure in a branch is too low, the corresponding will collapse.
\subsection {Environment\label{subsec:env}} 
The cellular environment consists of two diffusible
components regulating the behavior of the cells in various ways. Oxygen is evidently important for the cells to proceed through the cell cycle.
The local oxygen concentration is determined from the following equation:
\begin{equation}
\begin{aligned}
 \partial_t [O_2](x,t)&=\underbrace{D_{[O_2]}\nabla^2 [O_2](x,t)}_{\text{diffusion}} +\underbrace{\psi_{[O_2]}  n_v(x,t)([O_2]_\mathrm{blood}-[O_2](x,t))}_{\text{exchange with blood}}\\
&\quad+  \underbrace{[O_2](x,t]k_{[O_2]} \sum_{p=1}^Pn_p(x,t)}_{\text{Consumption}}
\end{aligned}
 \label{eq:reactdiff_ox}
\end{equation}
where $D_{[O_2]}$ is the diffusion coefficient of oxygen, $\psi_{[O_2]}$ denotes the
permeability of the oxygen through the vessels. $n_v(x,t)$ describes the surface area
occupied by the vessel at position $x$. $[O_2]_\mathrm{blood}(x,t)=[O_2]_\mathrm{ref}H(x,t)/H_\mathrm{in}$ defines the
oxygen concentration in a blood vessel located at position $x$. $[O_2]_\mathrm{ref}$ is a reference oxygen concentration, $H(x,t)$ is the haematocrit at location $x$ and time $t$ and $H_\mathrm{in}$ is the haematocrit at an inflow node and by default it is set to $H_\mathrm{in}=0.45$. The last term in~\eqref{eq:reactdiff_ox} reflects the fact that all cells
 consume oxygen at a
rate $k_{[O_2]}$.\\
A similar approach is used to describe the local concentration of growth factors
(e.g.,~\textit{Vascular Endothelial Growth Factors}) denoted by [VEGF]. The
latter is responsible for the growth of new blood vessels, which is especially
important for larger tumors. Initially, the tumor can benefit from the existing
vasculature, but when the tumor occupies a larger volume the oxygen supply
doesn't satisfy anymore and the cells are obliged to use their ability to
ask for new vessels, by secreting VEGF. Endothelial cells -- the building blocks
of blood vessels -- react and move chemotactically towards the hypoxic regions.
The corresponding reaction diffusion equation for VEGF reads:
\begin{equation}
\begin{aligned}
    \partial_t [VEGF](x,t)&=\underbrace{D_{\VEGF}\nabla^2[VEGF](x,t)}_{\text{diffusion}}
    -\underbrace{\psi_\VEGF n_v(x,t)[VEGF](x,t)}_{\text{exchange with
    blood}}+
  \underbrace{S_\VEGF(x,t)}_{\text{production}}\\
    &\quad-\underbrace{\delta_\VEGF \VEGF(x,t)}_{\text{decay}}\\
    S_{[\VEGF]}(x,t)& =k_{[\VEGF]} \sum_{p=1}^P\sum_{i=1}^{I_p(t)}
   \delta_{X_{i,p}(t)}\H([\VEGF_\mathrm{int}]_{i,p}(t) -
   [\VEGF_\mathrm{int}]_\mathrm{thr})
\end{aligned}
\end{equation}

 \begin{table}[h!]
\centering
\begin{tabular}{|c|c|c|c|}
\hline
Parameter & Oxygen  & VEGF& units \\\hline
$D$ & $\num{0.0014}$& $\num{6e-4}$ &\cmsq/min \\
$\psi$ & $6$  & $\num{6e-4}$ & cm/min\\
$\delta$ & $0$ &$ 0.6$& \sinv \\
$k_{[O_2]}$ & $-13$ & &\sinv\\
$k_{[\VEGF}]$&  &$0.6$&\sinv\\
$[O_2]_\mathrm{ref}$& $20$ & & mmHg\\
\hline
\end{tabular}
\caption{Parameter values reaction diffusion equations}
\label{tab:Diff}
\end{table}

\section{Variance reduction\label{sec:varred}}
In this section, we propose a variance reduction algorithm similar to the technique used in~\cite{Rousset2013} to
simulate bacterial chemotaxis. The main differences are due to the fact that (i) the model is not conservative; and (ii)
the internal dynamics only relates to cell division, apoptosis and VEGF secretion and not to advection-diffusive behavior.\\
As in~\cite{Rousset2013}, the algorithm relies on the combination of three simulations: a stochastic simulation with the full microscopic
model, as well as with a coarse approximation, combined with a deterministic grid-based simulation of the coarse model. The full microscopic
model uses an ensemble of $I_p(t)$ particles with state variables:
\begin{equation}
	\{ X_{i,p}(t),Z_{i,p}(t),\Phi_{i,p}(t),\zeta_{i,p}(t), [\q]_{i,p}(t),[\VEGF_\mathrm{int}]_{i,p}(t) \}_{i=1}^{I_p(t)}\label{eq:states_full}
\end{equation}
As the coarse agent-based model, we conceptually consider an agent-based model in which the internal state has been suppressed
and only the position remains:
\begin{equation}
	\{X_{i,p}^c(t)\}_{i=1}^{I_p(t)}\label{eq:states_control}
\end{equation}
So no internal dynamics is present, cells cannot divide, die or secrete VEGF. (In practice, we will use the results obtained from the full microscopic model, in which
we neglect apoptosis and  cell division, see later). The only dynamics is motion, which can be modeled with a PDE for the population density (see equation~\eqref{eq:macro}). We call
this coarse approximation the \textit{control process}. We also introduce the formal semigroup notation:
\begin{equation}
	e^{tL_p^c} \qquad\text{with } L_p^c(n_p^c)= -D_p\nabla^2-\chi_p\nabla\cdot\left[n_p(x,t)\left(1-\dfrac{n_p(x,t)}{n_{p,\max}}\nabla [VEGF](x,t)\right)\right]
\end{equation}
that represents the exact solution of the macroscopic partial differential equation~\eqref{eq:macro}. In practice, the solution will be approximated by a deterministic solution on a grid.
It should be clear that the advection-diffusion behaviour in both agent-based models is identical. Thus, the only difference between the two models occurs when cells divide or die. Assuming no reactions take place, the three processes thus have the same expectation.  This observation leads to the following variance reduction algorithm.
As an initial condition, we start from $I_p(0)$ particles sampled from specific
probability densities, resulting in the number density $n_p(x,0)$. For each particle, we choose a given internal state, for instance $\Phi_{i,p}(0)=Z_{i,p}(0)=0$, $\zeta_{i,p}(0)=0$, $[p53]_{i,p}(0)=0$,$[\VEGF_\mathrm{int}]_{i,p}(0)=0$, $1 \le i \le I_p(0)$, $1 \le p \le P$. (These internal states could also be sampled from an appropriate probability distribution.) Additionally, we introduce the variance reduced measure  $\bar{n}(x,t)$, which we initialize as $\bar{n}_p(x,0)=n_p(x,0)$. We denote the time step $\delta t$ and the discrete time instances $t^\ell=\ell\delta t$, $\ell=0,1,\ldots$
	\begin{alg}[Variance reduction for tumor growth]{\label{alg:varred_tumor}}
	We advance the variance reduced number density $\bar{n}(x,t)$ from time $t^{\ell}$ to $t^{\ell+1}$ as follows:
	\begin{itemize}
		\item Evolve the particle states \eqref{eq:states_full} from $t^\ell$ to $t^{\ell+1}$ using the agent-based model \eqref{eq:abm_model}.
		\item Compute the number density for the stochastic microscopic model using \eqref{eq:number_density}, as well as the number density for the coarse process as 
		\begin{equation}
			n_p^c(x,t^{\ell+1})=\sum_{i=1}^{I_p(t^\ell)}w_{i,p}(t^\ell)\delta_{X_{i,p}(t^{\ell+1})}
		\end{equation}
		i.e., we compute the number density for the control process based on particle positions and velocities at time $t^{\ell+1}$, taking into account only the particles that were present in the simulation at time $t^\ell$.
		\item Evolve the control number density $n^c_p(x,t)$ using a
		grid-based method based on \eqref{eq:macro} and add the reactions (the difference in number density due to cell division and apoptosis)
		\begin{equation}
			\bar{n}_p(x,t^{\ell+1}): = \bar{n}_p(x,t^\ell)\;e^{\delta t
			L^c} + n_p(x,t^{\ell+1})-n_p^c(x,t^{\ell+1}) \label{eq:npbar}
		\end{equation}  
\end{itemize}
\end{alg}

Next, we will prove that the proposed estimator for the population densities $n_p$ are unbiased and that the algorithm indeed reduces the variance on the mean population density. To this end, we define the so-called reaction field:
\begin{defn}[Reaction field] The control process differs from the full microscopic model in the way that there are no births, deaths or VEGF-secretion events. The direct influence on the population density can be summarized by the \textit{Reaction field} defined as:
\begin{equation}
\begin{aligned}
	R_p(x,t^{l+1}) &= n_p(x,t^{l+1}) - n_p^c(x,t^{l+1})\\
&= \sum_{i=1}^{I_p(t^{l+1})}w_{i,p}(t^{l+1})\delta_{X_{i,p}(t^{l+1}) }- \sum_{i=1}^{I_p(t^l)}w_{i,p}(t^{l})\delta_{X_{i,p}(t^{l+1})}\\
&= \sum_{i=1}^{I_p(t^l)}\left(w_{i,p}(t^{l+1})-w_{i,p}(t^{l})\right)\delta_{X_{i,p}(t^{l+1})}
\end{aligned}
\end{equation}
\end{defn}
\begin{defn}[Deterministic control density]
We also introduce a shorthand notation for the control density calculated with the macroscopic evolution equation:
\begin{equation}
	\tilde{n}(x,t^{l+1}) := \bar{n}(x,t)e^{\delta t L^c} \label{eq:nptilde}
\end{equation}
\end{defn}
This procedure will be repeated after reinitializing the control density $\tilde{n}_p(x,t)
=\bar{n}(x,t)$. The importance of reinitialization can be illustrated by looking
into the following hypothetical situation.  Suppose the $i$th cell of type $p$
divides at time $t=t^\star$, and hence cell $I_p+1$ is born. At time $t>
t^\star$, this newborn cell has moved randomly through the domain. Apart from
this random motion, it also has influenced the environment along its track.
Those events cannot be taken into account without reinitialization.
\begin{thm}[Unbiased estimator]
The algorithm described above yields an unbiased estimator for
the population density $n_p$.
\end{thm}
\begin{proof}
Assume that discretization errors are absent. Then, we  can calculate the expectation value of $\bar{n}_p$ based on equation~\eqref{eq:npbar} as follows:
\[
	\E\left[\bar{n}_p(x,t^{l+1})\right] = \E\left[\tilde{n}_p(x,t^{l+1})\right]+\E\left[n_p(x,t^{l+1})\right]-\E\left[n_p^c(x,t^{l+1})\right]
\]
By using the definition of sequentially the definition of $\tilde{n}_p^c$ (see equation~\eqref{eq:nptilde}) and the linearity of $\E$, we can conclude that $\E[\bar{n}_p]=\E[n_p]$ and hence $\bar{n}_p$ is indeed an unbiased estimator.
\end{proof}

\section{Results\label{sec:results}}
In this section, we will illustrate the performance of the variance algorithm described above with
various numerical experiments. The cells are living on a
 $50\times 50$ square grid. By default $2000$ normal cells are uniformly distributed over the whole domain.  A small tumor consisting of $200$ cancer cells are initially normally distributed with mean $0.25\Delta x$ and standard deviation $0.05\Delta x$ in close to the left vessel. We simulate the system over $1920$ timesteps (or 40 days). The whole set of default parameters is summarized in table~\ref{tab:poppar}.
 The normal tissue on the other hand is uniformly distributed over
 the whole domain. To initialize the agent-based simulation, we sampled
 $I_p(0)$ particles from the corresponding distribution. Remark that we mostly use a high number of particles to discretize the population density stochastically in order to make the agent-based model consistent with the continuum description. Hence, the equation to calculate the cell number density can be rewritten as:
\begin{equation}
	n_p(x,t) = \sum_{i=1}^{I_p(t)} q_{i,p} w_{i,p}(t) \delta_{X_{i,p}(t)}
\end{equation}

where an additional weight $q_{i,p}$ is attached to each particle. This implies that each particle has a lower mass. The total mass is $\sum_{i=1}^{I_p(t)}q_{i,p} w_{i,p}(t)$. During the numerical experiments, we will use weights $q_p$ independent of both the specific $i$-th particle of population $p$ and of the time.
\\
 Further, we initialize the environment as follows: two straight vessels at $x=20\Delta x$ and $x=40\Delta x$, corresponding to a
moderate vascular density of $50~\mathrm{cm}^2/\mathrm{cm}^3$
(see~\cite{Owen2011}). The latter results in average oxygen concentrations,
meaning that cells are proceeding through the cell cycle at a speed, which is
slightly higher than half maximal. More details concerning realistic vascular
densities and oxygen concentrations can be found in the supplementary material
provided with~\cite{Owen2011}.\\
The macroscopic equations are simulated using a simple Euler discretization for
the time derivative and a second order central finite volume to discretize the
spatial derivative. In the first three experiments, we have chosen for an
explicit method. Further, the linear systems originating from the reaction-diffusion PDEs
modeling the environment are solved using a conjugate gradient algorithm
(~\cite{eigenweb}). The default choice of discretization, time-step and number
of cells can be found in table~\ref{tab:poppar}.\\
The discussions corresponding to each of the individual experiments are
organized as follows. First we consider the evolution of the population
densities and the environment. Afterwards, we take a closer look at the variance with and without variance reduction. The results of the experiments are obtained by averaging
out over $100$ realizations.
\begin{rem}[Color code]
During the numerical experiments, we adopt the following color code to describe the different quantities:
\begin{itemize}
        \item A colormap from white (low) towards gray (high) is used for the mean population densities (both normal and cancer).
        \item A colormap from white (low) towards blue (high) is used for the variance on the mean cancer cell density (with and without variance reduction).
        \item An additional colormap from green (negative) towards white(zero) and blue(high) is used to denote the covariance between the density and the reaction field.
        \item A colormap from white (low) towards red (high) is used for the mean oxygen distribution.
        \end{itemize}
\end{rem}
\begin{rem}[units]
We use minutes as default time unit and cm as the spatial unit. Those will be omitted in the figure titles for compactness.
\end{rem}
\begin{table}[h!]
\centering
\begin{tabular}{|c|ccc|c|}
\hline
Parameter & Normal & Cancer & EC & units\\ \hline
$D_p$ & $0$ & $\num{5e-9}$ & $\num{1e-8}$ & \cmsq/min\\
$I_p(0)$ & $2000$ & $200$ & $0$ & $\#$particles\\
$q_{i,p}$ & $1$ & $0.5$ & $1$  & dimensionless\\
$\delta t$ & $30$ & $30$ & $30$ & min\\
$n_{p,\max}$ & $1$ & $2$ & $2$ & $\#$particles\\
$\Delta x$ & $\num{4e-3}$ & $\num{4e-3}$ & $\num{4e-3}$ & cm\\
$a$ & $0.5\Delta x$& $0.25\Delta x$& &cm \\
$b$  &$0.5\Delta x$ & $0.05\Delta x$& &cm\\
\hline
\end{tabular}
\caption{Default parameter set used for the numerical experiments\label{tab:poppar}}
\end{table}
 \subsection{A small-scale experiment with a static vasculature}
 \paragraph{Population densities}
 In figure~\ref{fig:EvoEnv}, we have plotted the population density of the normal and cancer tissue, along with the oxygen concentration at time $t=\num{5.76e5}$~min.
 \begin{figure}[h!]
\centering

    \subfigure{
 \includegraphics[width=1\textwidth]{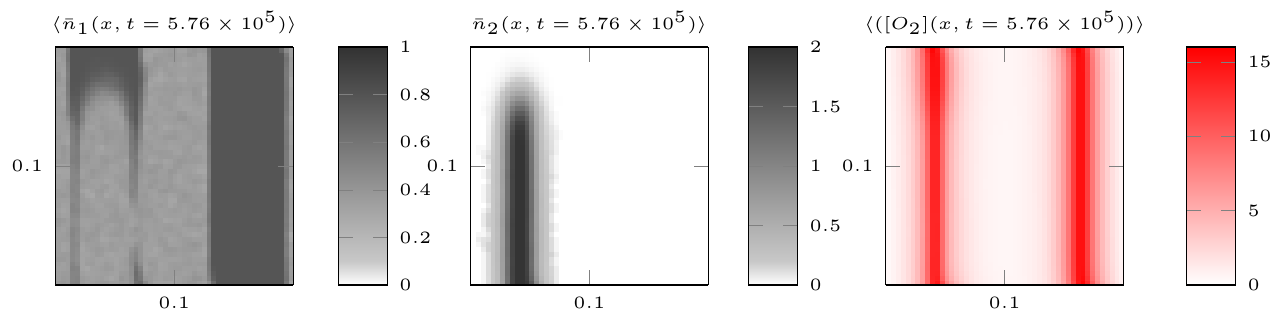}
}
     \caption{Mean cellular densities --normal tissue: (left panel), cancer cell
     density: (middle panel)-- and mean oxygen concentration (right panel) calculated
     using variance reduction at time~$t=\num{5.76e5}$~min.\label{fig:EvoEnv}}
\end{figure}

 The tumor immediately influences the normal
 tissue in the sense that a significant amount of normal cells die in this
 cancerous region. They literally have to make place for the growing tumor. Besides cells also die due to lack off oxygen
 in between the two existing vessels. In the meantime, the tumor starts to grow along the vessel until the tissue is
 locally saturated, meaning that $\sum_{p=1}^P n_p(x,t)>n_{p,\max}$. The high
 birth rate can be explained by the high oxygen concentration. Furthermore, the
 cells also diffuse in the other directions due to the random Brownian motion. This evolution
 can also be seen as an illustration of the ``go or grow''-paradigm, which is
 identified as an important characteristic of the aggressiveness of the
 tumor~\cite{Hatzikirou2012,Garay2013}.\\
 This process continues until the tissue is fully saturated along the
 leftmost vessel. Remark that none of the cells is able to cross the low oxygen zone. They would all die due to the hypoxic environment.

\paragraph{Evolution of the variance.}  We
 investigate how the population densities shown in
 figure~\ref{fig:EvoEnv} influence the variance. Figure~\ref{fig:EvoSmall}
 illustrates the elements contributing to the variance (see equation~\eqref{eq:varnp}). We first consider the variance
 with and without reduction in more detail. Combining the definitions of both variance and $n_p,\bar{n}_p$ yields:
 \begin{eqnarray}
 \Var[n_p(x,t)]&=& \Var[n_p^c(x,t)+R_p(x,t)]\\
	 &=& \Var[n_p^c(x,t)+\Var[R_p(x,t)]+2\mathrm{Cov}(n_p^c(x,t),R_p(x,t))\\
	\Var[\bar{n}_p] &=& \Var[\tilde{n}_p^c]+\Var[R_p(x,t)]+
	2\mathrm{Cov}(\tilde{n}_p^c(x,t),R_p)\label{eq:varnp}
\end{eqnarray}
 the variance on the reaction field (left), the variance on the corresponding
 control densities -- without variance reduction $n_2^c$ and with variance
 reduction $\tilde{n}_2^c$ -- (middle) and the covariance between the reactions
 and the control densities (right). We compare the results without
 (first row) and with variance reduction (second row).
  \begin{figure}[h!]
 \centering
\subfigure{
 	\includegraphics[width=\textwidth]{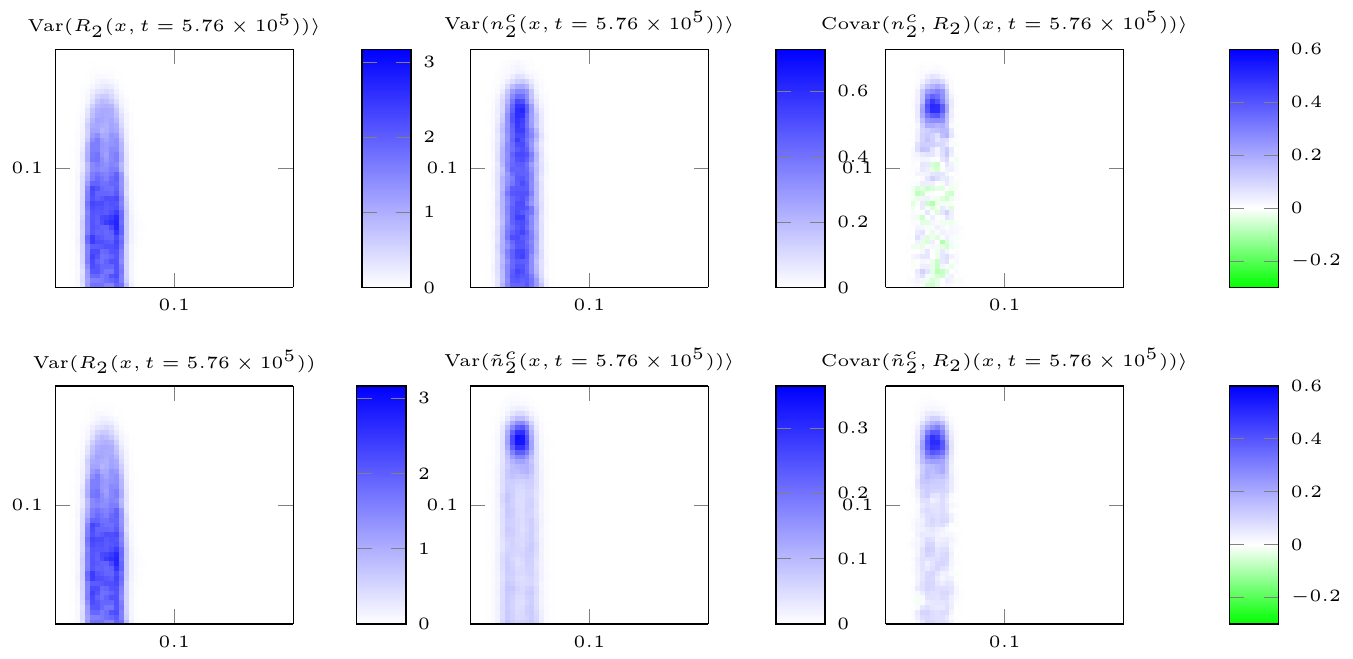}
       }
 \caption{Evolution of the factors contributing -- Var$(R_2(x,t))$
 (left), Var$(n_2^c(x,t))$ (middle) and
 Covar$((R_2,n_2^c)(x,t))$ (right) -- to variance on mean cancer cell density
 with (second row per time) and without variance reduction at time
 $t=\num{5.76e5}$~min.\label{fig:EvoSmall}}
 \end{figure}
 The variance on the reaction field is stretched along the leftmost vessel, as is the tumor itself. To explain the pattern in more detail, we
 have to compare the variance on the reaction field with the corresponding density. The variance
 is especially high just next to the largest concentration of the tumor, where the concentration of reactions  is high due
 to the combined effect of the relatively high number of cancer cells and the fact the tissue is not fully saturated yet.
 Remark that the variance on the
 reaction field does not depend on the variance reduction, since it is fully
 determined by the results of the agent-based simulation.\\
 However, the image is completely different for the variance on the
 corresponding control densities. Observing the middle picture on the first row leads to the conclusion 
 that the noise is dominated by random motion since the variance is
 larger in the middle of the tumor than at the border where most of the
 reactions take place. In contrast, after applying variance reduction the
 variance is fully determined by the reactions, implying that the
 variance is mostly filtered by the algorithm.\\
The above analysis (see figure~\ref{fig:EvoSmall}) of the evolution of the variance
 clearly demonstrates the strong correlation between the oxygen concentration and the variance on the
 population densities. Again a more detailed view of the evolution can be found in the supporting material~\ref{S2_Video}.

To illustrate the performance of the algorithm in another way, we have taken
some slices -- at $y=0.04~\mathrm{cm},y=0.1$~cm and $y=0.14~\mathrm{cm}$
respectively -- of the cancer cell density at $t=\num{5.76e5}$~min. In
figure~\ref{fig:slices} they are plotted along with their $95\%$ confidence
interval. The results based on the full stochastic model are plotted in red,
while the results using the variance reduction algorithm are colored in blue. It
can easily be seen that there is indeed a significant reduction and that the
results using the variance reduction algorithm are consistent with the original
results in the sense that they closely approximate the solution from the full stochastic model and that the variance is reduced significantly.
  \begin{figure}[h!]
     \centering
     
     \includegraphics[width=\textwidth]{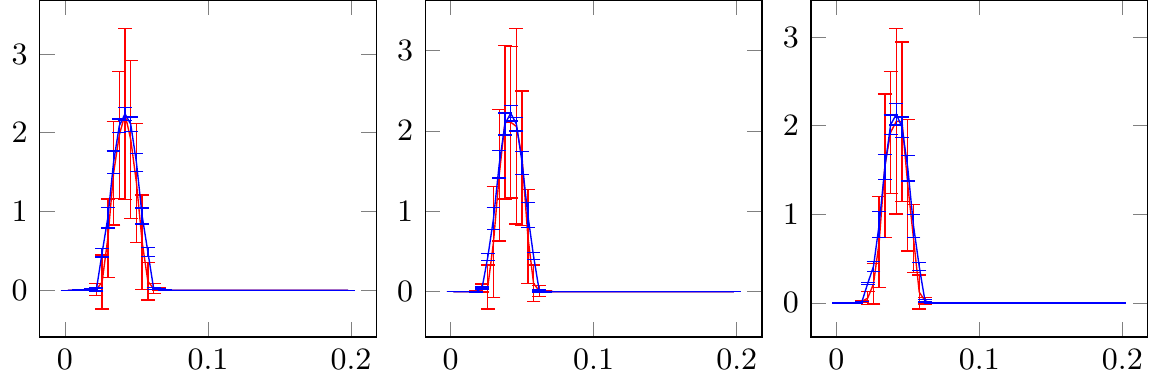}
     \caption{mean cancer cell density and reliability interval at time
     $t=\num{5.76e5}$min at $y=0.04$~cm (left), $y=0.1$~cm (middle) and $y=0.14$~cm
     (right) with (blue) and without variance reduction (red).
     \label{fig:slices}}
 \end{figure}

\subsection{Experiment on a larger domain}
As pointed out before, our lattice-free approach allows to rescale the system in
a straightforward way. Since the cost mainly depends on the number of particles and only
marginal on the the domain size, it is possible to consider to perform a
similar experiment on a rescaled (coarser) grid. To illustrate this we perform
the simulation with $\Delta x=\Delta y=\num{1.26e-2}$~cm, corresponding to a domain
of $0.4\mathrm{cm}^2$, corresponding with an upscaling of a factor $10$. The normal tissue initially consists of $\num{2e4}$
particles and a tumor of $1000$ cells. In
figure~\ref{fig:EvoPopLarge} we have plotted the evolution of the cellular
distributions of the different cell types and the corresponding oxygen
concentration. From the plot in the left column, one can see that normal cells are
multiplying along the rightmost vessel  since the left vessel is fully occupied by the tumor and 
there is not enough space for both the tumor and normal cells. In the rest of the domain the normal tissue is reduced to a minimal level due to lack of oxygen and the influence of the tumor.
 \begin{figure}[h!]
 \centering
	\includegraphics[width=1\textwidth]{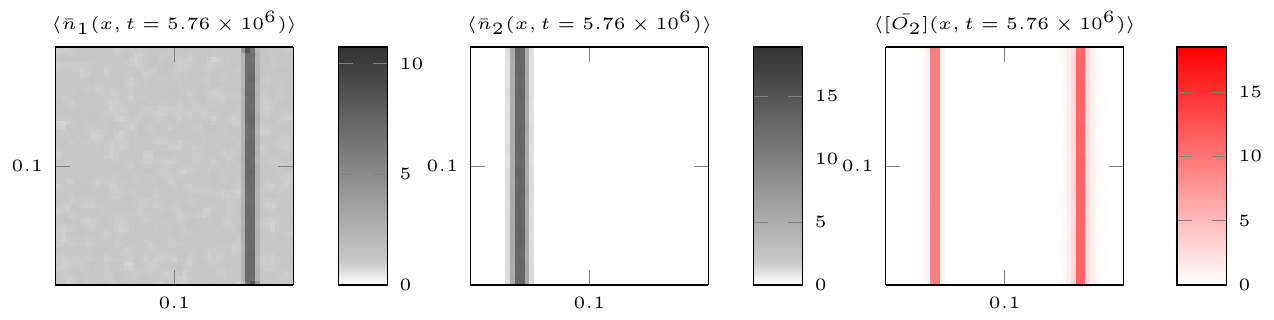}
	\caption{Mean population densities and oxygen distribution in a large scale
	setting at time $t=\num{5.76e6}$~min.\label{fig:EvoPopLarge}}
\end{figure}
In figure~\ref{fig:EvoVarLarge}, we examine the influence of the variance
reduction algorithm on the variance on the resulting tumor cell density as a
function of time. Comparing the variance plot with (right panel) and without (left panel) give rise to the observation that the algorithm again yields a reduction of the variance both in the center and at the border of the tumor. This implies that the border of the tumor can be estimated in a more accurate, which determines the harshness of the tumor.
 \begin{figure}[h!]
\centering
       \includegraphics[width=0.8\textwidth]{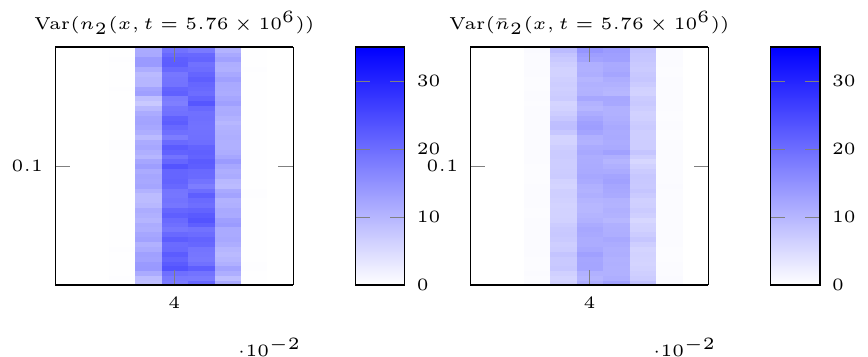}
%
\caption{Evolution of the variance on the mean cancer cell density in a
large-scale setting with and without variance reduction at
time $t=\num{5,76e6}$~min, zoomed in on the left vessel.\label{fig:EvoVarLarge}}
\end{figure}

\subsection{Variance reduction for sprouting angiogenesis}
As a last experiment we will examine the performance of the algorithm in the case
where the vasculature is also updated dynamically according to the model
outlined in the section~\ref{sec:models}. A small tumor mass of initially
$100$ cells - sampled from a uniform distribution, with parameters $a_2=0.3\Delta x,b_2=0.1\Delta x$. The population
density is discretized with $200$ cells, i.e. $q_{i,2}=0.5$. Further, we have chosen $D_2=\num{1e-8}$ cm\textsuperscript{2}/min. The other
parameters are set to the default values outlined in table~\ref{tab:poppar}.
 In contrast to
the previous experiments, the cancer cells are now able to cause extension of the vascular
network according to their needs. The resulting oxygen
distribution reveals a strong correlation with the cancer cell distribution itself, meaning that the
tumor is fully vascularized now and can grow further. A small fraction of the tumor even managed to reach the second vessel supported by some new branches in the vascular network created in response to the high VEGF gradients.
\begin{figure}[h]
\centering
        \includegraphics[width=1\textwidth]{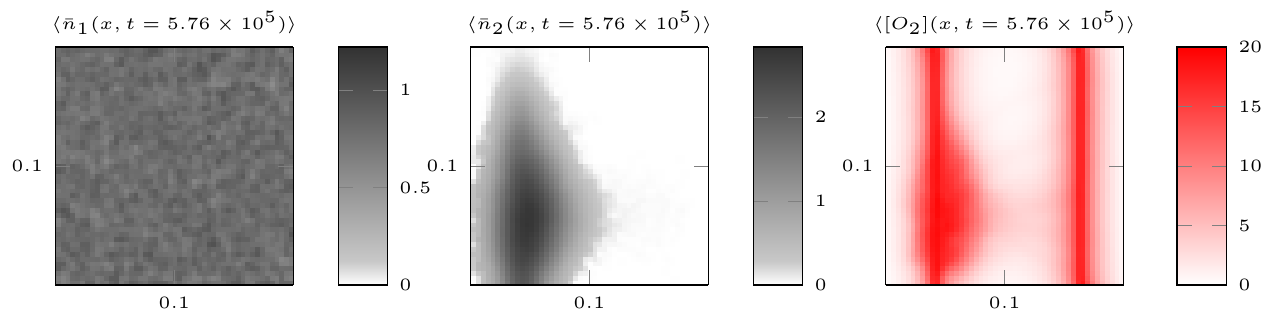}
\caption{Mean cellular densities and oxygen distribution at time $t=\num{5.76e5}$~min with dynamical vasculature\label{fig:evoPopAngio}}
\end{figure}
Next, we investigate how the variance reduction algorithm is performing in this
setting of dynamic vascularization. In figure~\ref{fig:VarAngio}, we have
plotted the variance on the mean cancer cell density with ($\bar{n_2}(x,t)$)
variance reduction on the right and and without variance reduction on the left.

\begin{figure}[h!]
\centering
\subfigure{
        \includegraphics[width=0.8\textwidth]{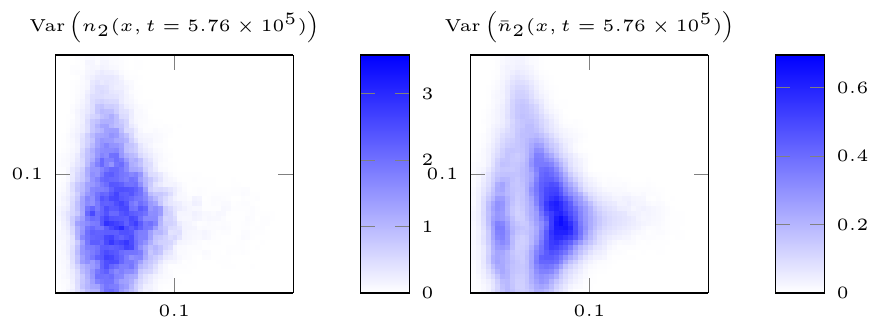}
      }
      \caption{Variance with and without variance
      reduction on the mean cancer cell density with a dynamic vasculature at time $t=\num{5.76e5}$~min\label{fig:VarAngio}}
\end{figure}
A comparison of the two variance plots in figure~\ref{fig:VarAngio} leads to the
conclusion that the variance is reduced everywhere. Along the leftmost vessel,
the algorithm was even able to eliminate the noise completely. Indeed,the tissue
is saturated here, so new cells are born here. Just outside this zone of maximal
saturation the tissue is not so dense giving rise to
more births and a higher level of noise here. In this region, the noise is
proportional to the local density itself.

\subsection{Fast diffusing cancer cells}
 Motivated by the hypothesis
that the diffusion coefficient can be related to the aggressiveness of the
tumor, we investigate the situation where cancer cells have a higher diffusion
coefficient. Swanson et al.
have shown~\cite{Swanson2008} that glioma's with a higher diffusion coefficient have a higher probability to cause metastases, which is
obviously an important characteristic for the long-term survival probability of
the patient. Apart from the modified diffusion coefficient
$D_2=\num{5e-7}$~cm\textsuperscript{2}/min, we adopt the same initial
configuration as in the previous experiment. Remark that the simulations are performed with a smaller time-step ($\delta t=0.3$~min) in order to fulfill
the CFL-condition corresponding to the macroscopic equation.
\paragraph{Populations}
In figure~\ref{fig:HighDiffpop}, the evolution of both the normal tissue and the
tumor are shown along with the local oxygen concentration at $t=\num{1.152e3}$~min.
As in the previous experiment, the normal tissue density is the result of the cell deaths due to the presence of 
the tumor, while the normal cells are more sensitive to hypoxic environment. 
The tumor, on the other hand, has diffused through
 the normal tissue significantly on this short timescale, without consuming too much oxygen.
 Indeed, the cancer cells have already covered a large distance within a
 rather short time interval, meaning that the tumor exhibits the
 \textit{go}-phenotype, rather than the \textit{grow}-phenotype as it was the
 case in the first experiment. Despite the fact that the tumor didn't cause a
 lot of damage, it is potentially dangerous since it can stay more or less invisible for a long time and as soon as the tumor reaches a vessel it is possible that cells invade a vessel and give rise to metastatic spread of the cancer. 
 \begin{figure}[h!]
\centering
        \includegraphics[width=1\textwidth]{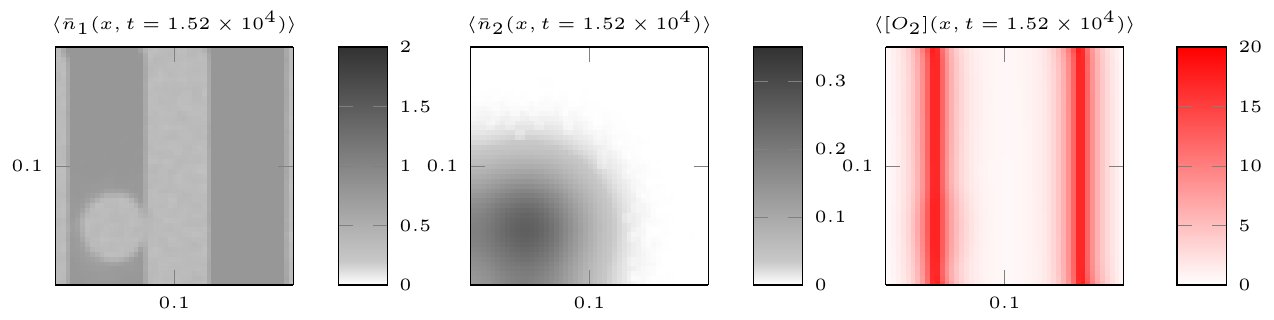}
	\caption{Evolution of the cellular populations as a
	function of time (at time $t=\num{1.52e4}$~min, with fast diffusing cancer cells.\label{fig:HighDiffpop}}
\end{figure}

\paragraph{Evolution of variance}
As before, we also examine the variance on the mean cancer cell density with
and without variance reduction. Without variance
reduction, the resulting variance is proportional to the density itself,
suggesting that the variance is mainly caused by the random jumps. Obviously,
the latter will be higher in zones with more cells. When variance reduction is
applied, the variance is reduced with at least a factor $100$ pointwise and
moreover the plot reveals a clear pattern, which is again related to the oxygen
concentration.
 \begin{figure}[h!]
	\centering
 	\subfigure{
 	        \includegraphics[width=0.8\textwidth]{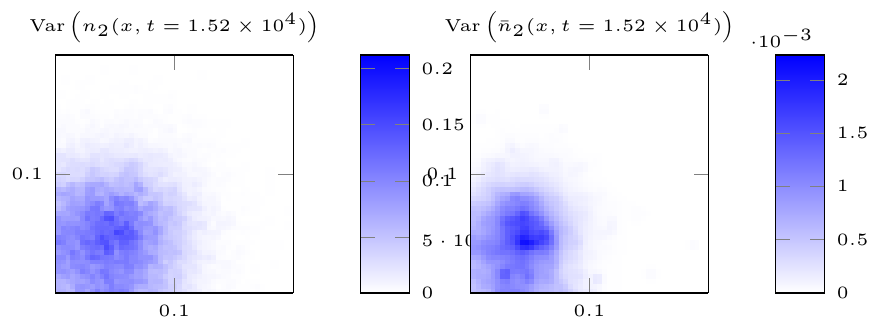}
%
     }

  	  \caption{Evolution of the variance on the mean cancer cell density with
       and without variance reduction at time $t=1152$~min.}
\end{figure}

\section{Discussion\label{sec:disc}}
We developed a novel variance reduction technique
specifically suited to reduce the noise of agent-based models with birth and death events, as it is the case
in our model for tumor growth. We proved that the algorithm outlined
in section~\ref{sec:varred} gave rise to an unbiased estimator and the
variance is determined by the births and deaths. The performance was illustrated
numerically in different possible regimes characterizing different aspects of
tumor growth such as sprouting angiogenesis, highly diffusive cancer cells and
large-scale systems.
The proposed algorithm is based on the idea of control variates, since the
evolution of the system without reactions is known deterministically via the
macroscopic equation~\eqref{eq:macro}.\\A valuable extension would be to combine
this algorithm with other variance reduction techniques such as importance sampling.
It is self-evident that an accurate and efficient simulation of
all the different aspects of the system is crucially important for the
reliability of the system as a whole. Apart from that, we will also extend
our model with important features such as haptotaxis in response to the extra-cellular matrix and include a more
sophisticated model for
stemcellness~\cite{Cicalese2009,Morrison2006,Colaluca2008} since it was
identified as one of the hallmarks of
cancer~\cite{Hanahan2011}. Another track worthwhile further
investigation is to apply our technique to patient-specific data. For instance,
patient-specific data, like MRI-images or blood parameters, could be used as a
specific initial configuration~\cite{Macklin2012}.
Finally, this algorithm can also be applied on related systems such as bone fracture healing and other 
application where we are interested in macroscopic behavior, but with agent-based features characterizing the dynamics.
\section{Supporting Information}
\subsection{S1 Video}
\label{S1_Video}
{\bf Evolution of the population densities in the small scale setting.} Evolution of the mean normal and cancer cell density from time $t=0$ till time $t=\num{5.76e5}~\min$. The initial configuration corresponds with the first numerical experiment.
\subsection{S2 Video}
\label{S2_Video}
{\bf Evolution of the variance in the small scale setting.} Evolution of the variance on the mean normal and cancer cell density from time $t=0$ till time $t=\num{5.76e5}$~min.

\bibliographystyle{abbrv}
\bibliography{references}

\end{document}